\documentclass[preprint,11pt]{elsarticle}
\renewcommand{\baselinestretch}{1.5} 
\usepackage{geometry}
\usepackage{array,xspace,multirow,hhline,tikz,colortbl,tabularx,booktabs,fixltx2e,amsmath,amssymb,amsfonts,amsthm}
\usetikzlibrary{shapes}
\usepackage{algorithm}
\usepackage{algorithmic}
\usepackage{eqparbox}
\usepackage{verbatim,ifthen}
\usepackage{enumitem}
\usepackage{pifont}
\usepackage{ifthen}
\usepackage{calrsfs,mathrsfs}
\usepackage{bbding,pifont}
\usepackage{pgflibraryshapes}
\usepackage{varioref}
\usepackage{subfigure}
\usepackage{url}
\usepackage{nicefrac}

\definecolor{light-gray}{gray}{0.9}

\newtheorem{theorem}{Theorem}[section]

\newtheorem{proposition}[theorem]{Proposition}
\newtheorem{corollary}[theorem]{Corollary}

\newtheorem{example}[theorem]{Example}

\newtheorem{open}{Open Problem}

\theoremstyle{definition}
\newtheorem{remark}[theorem]{Remark}

	\newcommand\eat[1]{}

	\usepackage{enumitem}
	\setenumerate[1]{label=\rm(\it{\roman{*}}\rm),ref=({\it\roman{*}}),leftmargin=*}
	\newlength{\wordlength}

	\newcommand{\eqclass}[2][]{\ifthenelse{\equal{#1}{}}{[#2]}{[#2]_{\sim_{#1}}}}

\usepackage{enumitem}
\setenumerate[1]{label=\rm(\it{\roman{*}}\rm),ref=({\it\roman{*}}),leftmargin=*}

\newcommand{\nbh}[1][]{
	\ifthenelse{\equal{#1}{}}{\nu}{\nu(#1)}
}

\newcommand{\cstr}[1][]{
	\ifthenelse{\equal{#1}{}}{\mathscr S}{\cstr(#1)}
}

\newcommand{\choice}[1][]{
\ifthenelse{\equal{#1}{}}{\mathit{C}}{\choice(#1)}

\newcommand{\ml}[1][]{\ensuremath{\ifthenelse{\equal{#1}{}}{\mathit{ML}}{\mathit{ML}(#1)}}\xspace}
\newcommand{\sml}[1][]{\ensuremath{\ifthenelse{\equal{#1}{}}{\mathit{SML}}{\mathit{SML}(#1)}}\xspace}
\newcommand{\sd}[1][]{\ensuremath{\ifthenelse{\equal{#1}{}}{\mathit{SD}}{\mathit{SD}(#1)}}\xspace}
\newcommand{\rsd}[1][]{\ensuremath{\ifthenelse{\equal{#1}{}}{\mathit{RSD}}{\mathit{RSD}(#1)}}\xspace}
\newcommand{\rd}[1][]{\ensuremath{\ifthenelse{\equal{#1}{}}{\mathit{RD}}{\mathit{RD}(#1)}}\xspace}
\newcommand{\st}[1][]{\ensuremath{\ifthenelse{\equal{#1}{}}{\mathit{ST}}{\mathit{ST}(#1)}}\xspace}
\newcommand{\bd}[1][]{\ensuremath{\ifthenelse{\equal{#1}{}}{\mathit{BD}}{\mathit{BD}(#1)}}\xspace}
\newcommand{\pc}[1][]{\ensuremath{\ifthenelse{\equal{#1}{}}{\mathit{PC}}{\mathit{PC}(#1)}}\xspace}
\newcommand{\dl}[1][]{\ensuremath{\ifthenelse{\equal{#1}{}}{\mathit{DL}}{\mathit{DL}(#1)}}\xspace}
\newcommand{\ul}[1][]{\ensuremath{\ifthenelse{\equal{#1}{}}{\mathit{UL}}{\mathit{UL}(#1)}}\xspace}

	\newcommand{\indiff}{\ensuremath{\sim}}}
	
	\usepackage{boxedminipage}
	\usepackage{xspace}

\newcommand{\pbDef}[3]{%
\noindent
\begin{center}
\begin{boxedminipage}{0.98 \columnwidth}
#1\\[5pt]
\begin{tabular}{l p{0.75 \columnwidth}}
Input: & #2\\
Question: & #3
\end{tabular}
\end{boxedminipage}
\end{center}
}

\graphicspath{{.}}

\begin{document}

\title{Computing Welfare-Maximizing Fair Allocations of Indivisible Goods}
		
\author{Haris Aziz}\ead{haris.aziz@unsw.edu.au}
\address{UNSW Sydney, Kensington, Sydney, Australia, NSW 2035}

\author{Xin Huang}\ead{xinhuang@campus.technion.ac.il }
\address{Technion, Haifa 32000, Israel}

\author{Nicholas Mattei}\ead{nsmattei@tulane.edu}
\address{Tulane University, 6823 St Charles Ave, New Orleans, LA 70118}

\author{Erel Segal-Halevi (corresponding author) }
\ead{erelsgl@gmail.com}
\address{Ariel University, Kiriat Hamada 3, Ariel 40700, Israel}


\begin{abstract}
We analyze the run-time complexity of computing allocations that are both fair and maximize the utilitarian social welfare, defined as the sum of agents' utilities. We focus on two tractable fairness concepts: envy-freeness up to one item (EF1) and proportionality up to one item (PROP1). We consider two computational problems: (1) Among the utilitarian-maximal allocations, decide whether there exists one that is also fair; (2) among the fair allocations, compute one that maximizes the utilitarian welfare. We show that both problems are strongly NP-hard when the number of agents is variable, and remain NP-hard for a fixed number of agents greater than two. For the special case of two agents, we find that problem (1) is polynomial-time solvable, while problem (2) remains NP-hard. Finally, with a fixed number of agents, we design pseudopolynomial-time algorithms for both problems. 
We extend our results to the stronger fairness notions envy-freeness up to any item (EFx) and proportionality up to any item (PROPx).
\end{abstract}

\maketitle


\sloppy

\newpage

\section{Introduction}
There are many problems in which both \emph{fairness} and \emph{efficiency} are important considerations. Recent examples from the operations research literature are scheduling \cite{agnetis2019price}, disaster relief \cite{erbeyouglu2020robust}, vehicle routing \cite{jozefowiez2008multi} ambulance planning \cite{jagtenberg2020improving}, and multi-portfolio optimization \cite{iancu2014fairness}.
In this paper we focus on algorithms for \emph{allocating indivisible goods} among agents. 
Such algorithms have broad impact in a number of areas including 
{school choice \cite{APR05a}, conference paper assignment \cite{LMNW18a}, course allocation \cite{BuCa12a}, warehouse delivery \cite{karaenke2020non},} and many others.  Two often competing objectives are balancing the \emph{welfare} of the allocation, defined as the sum of the utilities of the agents, with the \emph{fairness}, which concerns the utility of each individual agent.

When allocating indivisible items, perfect fairness may be unattainable even when there are two agents and a single item. Indeed, many algorithms for fair allocation of indivisible items simply fail if a fair allocation does not exist \cite{kilgour2018two}.
An alternative approach, which is arguably more suitable than just failing, is \emph{fairness up-to one item}. For example, \emph{envy-freeness up to one item (EF1)} requires that, for any pair of agents, if at most one item is removed from one agent's bundle, then the other agent does not envy \citep{Budi11a}.
{There are many other fairness notions based on the ``up-to one item'' concept; see Section \ref{sec:setting} for the formal definitions.}

Fairness requirements are often complemented by requirements for economic efficiency, the most common of which is \emph{Pareto-efficiency (PE)}. Many recent works have studied PE+fair allocations for various fairness notions (see Related Work). 
However, PE alone is not a sufficient condition for economic efficiency. As an example, evidence from course allocation shows that Pareto-efficient mechanisms perform poorly on natural measures such as the fraction of students getting their first choice, or the average rank of a student \citep{BuCa12a}.
Such measures can be captured by a stronger measure of economic efficiency: the \emph{utilitarian welfare}, defined as the sum of all agents' utilities, or equivalently, the average utility per agent
\cite{Moul03a}.

There are other settings in which the sum of utilities is a natural measure of efficiency. For example, if the items for allocation are vaccinations, and the utility of an agent is proportional to their probability to survive given the allocated vaccination (which can be computed using statistics on past medical records), then the utilitarian welfare corresponds to the expected number of survivors. 
If the items are allocated by a politician who wants to be re-elected, and the utility of an agent determines the probability that they vote for the politician, then the utilitarian welfare corresponds to the expected number of supporters. 
Finally, in situations of repeated allocations, e.g. when different items are allocated among the same agents each day, an algorithm attaining a higher average utility per agent may eventually lead to a higher total utility for every agent.
To illustrate, suppose there are two agents and two possible allocations: one gives the first agent utility 3 and the second agent utility 1, while the other gives the first agent utility 1 and the second agent utility 5. Both allocation are Pareto-efficient. However, if the setting repeats daily, with the roles of ``first agent'' and ``second agent'' selected uniformly at random each day, then the second allocation is superior, since it gives both agents an average utility of 3 per day.
{
Some other applications in the operations research literature, in which utilitarian welfare is used, are house allocation \cite{arnosti2020design}, school choice \cite{biro2021complexity} and product-line design \cite{kohli1990heuristics}.
}

We refer to allocations with a highest utilitarian welfare as \emph{utilitarian maximal (UM)}. 
We focus on the complexity of computing an allocation that maximizes the sum of utilities among those that satisfy a given (approximate) fairness notion. We also consider the problem of deciding whether an allocation exists which simultaneously maximizes the sum of utilities and satisfies a fairness notion.  These results shed light on the settings where we are able to guarantee the tractable computation of allocations that are both fair and efficient.
It is well-known that any UM allocation is PE, but the opposite is not true. So the combination of UM and fairness is strictly stronger than PE and fairness.

\smallskip
\textbf{Contributions.} Given a fairness requirement, we want to decide whether there exists a UM allocation that satisfies it. If no such allocation exists, we want to find an allocation with highest welfare among the fair ones. However, we show that both these goals are NP-hard for three or more agents even when the number of agents is fixed, and strongly NP-hard when the number of agents is variable, that is, part of the input.\footnote{By strongly NP-hard, we mean that the problem remains NP-hard even if the numbers in the input are represented in unary representation~\citep{GaJo78a}
.}

When there are only two agents, deciding the existence of fair and UM allocations turns out to be polynomial-time solvable. 
In contrast, welfare maximization within the set of fair allocations is NP-hard even for two agents. 
Finally, for any fixed number of agents, we present a pseudopolynomial time algorithm for maximizing welfare constrained to fair allocations.
Hence, we obtain a clear understanding of the complexity of efficient fair allocation w.r.t. the number of agents. 
Our results are summarized in Table~\ref{table:summary-complexity}; see Section \ref{sec:setting} for the formal definitions of the fairness concepts.

\begin{table*}[h!]
\begin{center}
\scalebox{0.7}{
\begin{tabular}{  l  l  ll }
\toprule
\toprule
&$\mathbf{n=2}$ &\textbf{Fixed} $\mathbf{n\geq 3}$& \textbf{Arbitrary $n$ }\\ 
\midrule
\midrule

\textbf{Exists UM and EF1} &in P~(Th.~\ref{2agents-ef1-um})&
\shortstack{NP-complete~(Th.~\ref{um-ef1-hard}),\\pseudo-poly.~(Th.~\ref{um-within-ef1-pseudopoly})}
&
\shortstack{
strongly NP-complete
\\
(Th.~\ref{um-ef1-stronghard})
}
\\ \midrule

\textbf{Exists UM and PROP1} &in P~(Th.~\ref{2agents-prop1-um})&
\shortstack{
NP-complete~(Th.~\ref{um-prop1-hard}),\\
pseudo-poly.~(Th.~\ref{um-within-prop1-pseudopoly})}
&
\shortstack{
strongly NP-complete
\\
(Th.~\ref{um-prop1-stronghard})
} \\ \midrule

\textbf{Compute UM within EF1} &
 NP-hard~(BGJKN 4)* & 
 \shortstack{
 NP-hard~(Cor.~\ref{um-within-ef1-hard}),\\
pseudo-poly.~(Th.~\ref{um-within-ef1-pseudopoly})}
&
\shortstack {
strongly NP-hard 
\\(BGJKN 5;
Cor.~\ref{um-within-ef1-strongly-hard}
}
)\\ 
\midrule
\textbf{Compute UM within PROP1}&
 NP-hard~(Th.~\ref{th:um-within-prop1})* &
 \shortstack{NP-hard (Cor.~\ref{um-within-prop1-hard}),\\
 pseudo-poly.~(Th.~\ref{um-within-prop1-pseudopoly})}.
 &
 \shortstack{
 strongly NP-hard
 \\ (Cor.~\ref{um-within-prop1-strongly-hard})
 }\\
 
\midrule
\midrule

\textbf{Exists UM and EFx} &
NP-complete (Th.~\ref{2agents-efx-um})&
\shortstack{NP-complete~(Th.~\ref{um-efx-hard}),\\pseudo-poly.~(Th.~\ref{um-within-efx-pseudopoly})}
&
\shortstack{
strongly NP-complete
\\
(Th.~\ref{um-efx-stronghard})
}
\\ \midrule

\textbf{Exists UM and PROPx} &NP-complete~(Th.~\ref{2agents-propx-um})&
\shortstack{
NP-complete~(Th.~\ref{um-propx-hard}),\\
pseudo-poly.~(Th.~\ref{um-within-propx-pseudopoly})}
&
\shortstack{
strongly NP-complete
\\
(Th.~\ref{um-propx-stronghard})
} \\ \midrule

\textbf{Compute UM within EFx} &
 NP-hard~(Cor. \ref{2agents-um-within-efx}) & 
 \shortstack{
 NP-hard~(Cor.~\ref{um-within-efx-hard}),\\
pseudo-poly.~(Th.~\ref{um-within-efx-pseudopoly})}
&
\shortstack {
strongly NP-hard 
\\(Cor.~\ref{um-within-efx-strongly-hard})
}\\ 
\midrule
\textbf{Compute UM within PROPx}&
 NP-hard~(Cor. \ref{2agents-um-within-propx}) &
 \shortstack{NP-hard (Cor.~\ref{um-within-propx-hard}),\\
 pseudo-poly.~(Th.~\ref{um-within-propx-pseudopoly})}.
 &
 \shortstack{
 strongly NP-hard
 \\ (Cor.~\ref{um-within-propx-strongly-hard})
 }\\
\bottomrule
\bottomrule
\end{tabular}
}
\end{center}
\caption{Complexity of existence and computation of allocations that are welfare maximizing and fair.  
See Section \ref{sec:setting} for the formal definitions of the various fairness concepts.
BGJKN $k$ refers to a result that is implied by Theorem $k$ from \citet{BJK+19b}. An asterisk (*) means that the hardness proof uses non-normalized valuations (the sum of valuations of one agent is larger than the sum of valuations of the other agent); it is open whether the problem remains hard in the common special case of normalized valuations (see Question \ref{open:normalized}).
}
\label{table:summary-complexity}
\vspace*{-4mm}
\end{table*}

\section{Setting}
\label{sec:setting}
An allocation problem is a tuple $(N,O,u)$ such that $N=\{1,\ldots, n\}$ is a set of agents, $O=\{o_1,\ldots, o_m\}$ is a set of items and 
$u$ specifies a utility function $u_i: O \rightarrow \mathbb{R^+}$ for each agent $i\in N$. We assume agents have additive utility. That is, $u_i(O')=\sum_{o\in O'}u_i(o)$ {for every subset $O'\subseteq O$.}

An {\em allocation} $p$ is a function $p:N\rightarrow 2^O$ assigning each agent a set of items. Allocations must be complete, i.e., all items are allocated, $\bigcup_{i\in N}p(i)=O$, and the bundles of items assigned to agents must be disjoint, i.e., no two agents can be assigned the same items, $p(i)\cap p(j)=\emptyset$ for all $i,j\in N$.
For a given instance, we denote by $\mathcal{A}$ the set of all allocations. 
We do not consider strategic manipulations --- we assume that all agents truthfully report their valuations ---  leaving strategic issues for future work \cite{BCM15a}.

\subsection{Fairness}
An allocation $p$ is called:
\begin{itemize}
\item 
\emph{Proportional (PROP)} if for each agent $i\in N$, $u_i(p(i))\geq u_{i}(O)/n$.
\item 
\emph{Proportional up to $c$ items (PROP$c$)} if for each agent $i\in N$,
there exists a subset $O'\subseteq O\setminus p(i)$ of cardinality $\leq c$ for which
$u_i(p(i) \cup O')\geq u_{i}(O)/n$. 
\item 
{
\emph{Proportional up to any item (PROPx)} if for each agent $i\in N$,
\emph{for all} subsets $O'\subseteq O\setminus p(i)$ of cardinality $1$, it holds that
$u_i(p(i) \cup O')\geq u_{i}(O)/n$. }
\item 
\emph{Envy-free (EF)}  if 
$u_i(p(i))\geq u_i(p(j))$ for all $i,j\in N$. 
\item 
\emph{Envy-free up to $c$ items (EF$c$)} if for all $i,j \in N$, there exists a subset $O'\subseteq p(j)$ of cardinality $\leq c$ such that $u_i(p(i)) \geq u_i(p(j)\setminus O')$. 
\item 
{
\emph{Envy-free up to any item (EFx)} if for all $i,j \in N$, \emph{for all} subsets $O'\subseteq p(j)$ of cardinality $1$, it holds that  $u_i(p(i)) \geq u_i(p(j)\setminus O')$. }
\item 
\emph{Equitable (EQ)} if for all $i,j \in N$, $u_i(p(i)) = u_{j}(p(j))$.
\item 
\emph{Equitable up to $c$ items (EQ$c$)} if for all $i,j \in N$,
there exists a subset $O'\subseteq p(j)$ of cardinality $\leq c$ for which
$u_i(p(i))\geq u_{j}(p(j)\setminus O')$. 
\item 
{
\emph{Equitable up to any item (EQx)} if for all $i,j \in N$,
\emph{for all} subsets $O'\subseteq p(j)$ of cardinality $1$, it holds that
$u_i(p(i))\geq u_{j}(p(j)\setminus O')$. 
}
\end{itemize}
{Note the key differences between the fairness notions: the proportionality-based and envy-free-based notions only compare valuations of the same agent, while equitability-based notions compare valuations of different agents. 
We consider intra-agent utility comparisons to be more meaningful than inter-agent utility comparisons; therefore, we focus on fairness notions based on proportionality and envy-freeness in the present paper.
}

It is well-known that, with additive valuations, EF implies EF1 and PROP.
Moreover, 
EF1 implies PROP1,
but PROP1 is strictly weaker than EF1 even for two agents
(see Appendix \ref{sec:ef1-prop1}).

For a given instance, {and a given fairness notion F (EF1, PROP1, EFx, PROPx, etc.)}, we denote by $\mathcal{A}^{\text{F}}$ the sets of all F allocations.

\subsection{Welfare}
While there are multiple notions of welfare, we focus on \emph{utilitarian-maximality}.
Allocation $p$ is:
\begin{itemize}
\item \emph{Utilitarian-maximal (UM)} if it maximizes the sum of utilities:
$$
p\in \arg\max_{q\in \mathcal{A}} \sum_{i\in N} u_i(q(i)).
$$
\item 
\emph{Utilitarian-maximal (UM) within F}, for a given fairness notion F, if 
$$
p\in \arg\max_{q\in \mathcal{A}^{\text{F}}}\sum_{i\in N} u_i(q(i)).
$$
\end{itemize}

\section{Related work}
\citet{BCM15a} present a general survey of the main algorithms and considerations in fair item allocation from a computer science perspective.
\citet{karsu2015inequity} present a more focused survey of the tradeoff between efficiency and fairness in the operations research literature.

It is well-known that an EF1 and PROP1 allocation can be computed in polynomial time \citep{LMMS04a}.
Similarly, a UM allocation can be computed in polynomial-time by just giving each item to an agent whose value for the item is highest. 
Utilitarian welfare can be maximized even with simple sequential mechanisms \citep{bouveret2011general,kalinowski2013social}.
Our results in Table \ref{table:summary-complexity} show that the combination of tractable fairness and tractable welfare requirements may be intractable.

A Pareto-efficient and PROP1 allocation can be computed in strongly-polynomial time in various settings \citep{CFS17a,ACI+18,BrSa19a,barman2019proximity,aziz2020polynomial}. The complexity of computing a Pareto-efficient and EF1 allocation is an open question, but a pseudopolynomial time algorithm is known \citep{BMV17a}.
Our results show that strengthening Pareto-efficiency to utilitarian welfare-maximization leads to strong NP-hardness.

\citet{bouveret2008efficiency,KBKZ09b,bliem2016complexity} study the computational complexity of finding an allocation that is both Pareto-efficient and envy-free. In their future work, \citet{bliem2016complexity} mention that ``a different theoretical route would be to extend the investigations also to ... approximate envy-freeness'', which is our focus.

\citet{bredereck2019high} present a meta-algorithm that
 can find efficient and fair allocations for various notions of efficiency and fairness. 
 Among others, their algorithm can handle notions of \emph{group Pareto-efficiency} \citep{aleksandrov2018group}, one of which is equivalent to UM. 
 However, the runtime of their algorithm is very large: it is larger than $d^{2.5 d}$, where $d$ is the number of variables in the resulting integer linear program (see their Proposition 8). This $d$ is larger than $((4 n V)^n)^{m(n+1)}$, where $n$ is the number of agents, $m$ the number of item-types, and $V$ the largest value of an item (see the end of their subsection 4.3).   In other words, their runtime is doubly-exponential in $m$ and $n$, and singly-exponential in $V$. Accordingly, they note in their conclusion section that their ILP solution is mainly a ``classification result'', and note that ``this leads us to the open questions of providing an algorithm for efficient envy-free allocation with better running time, or running-time lower bounds''.  In contrast, our algorithms (in Section \ref{sec:fixedn}) run in time singly-exponential in $n$, and polynomial in $m$ and $V$, addressing their open question.%
\footnote{
 Simultaneously to the present work, 
 \citet{bredereck2021high} have also improved the practical applicability of their ILP-based approach.
}

In previous work, \citet{BJK+19b}[Theorems 4, 5] proved hardness of problems that they call FA-EF1 and HET-EF1. Their results imply that computing a UM within EF1 allocation is NP-hard for any fixed number $n\geq 2$ of agents, and strongly NP-hard when $n$ is arbitrary (unbounded).
Their results do not cover PROP1, nor the problems of whether a UM and fair allocation exists (see Table \ref{table:summary-complexity} for comparison). In contrast to our hardness results, \citet{benabbou2020finding} showed that when valuations are submodular with binary marginals (each item adds value $0$ or $1$ to each bundle),  UM+EF1 allocations exist and can be found efficiently. 

\citet{freeman2019equitable} study the computation of allocations that are both PE and EQ1, as well as a stronger notion that they call EQx. 
They prove that, when all utilities are strictly positive, then a PE and EQx allocation always exists, and a PE+EQ1 allocation can be found in pseudopolynomial time. However, when some utilities may be zero, deciding whether a PE and EQ1 / EQx / equitable allocation exists is strongly NP-hard.
They do not discuss utilitarian-welfare maximization.


Many recent works aim to maximize the \emph{Nash welfare} --- the product of utilities. This problem is NP-hard, but various approximations are known \cite{cole2015approximating,darmann2015maximizing,branzei2017nash,caragiannis2019envy,caragiannis2019unreasonable,amanatidis2020maximum}.
We focus on the \emph{sum} of utilities,
which is one of the standard ways to measure the total welfare in society.

The fairness-welfare tradeoff has also been studied through the lens of the \emph{price of fairness} --- the ratio between the maximum welfare of an arbitrary allocation and the maximum welfare of a fair allocation.
Bounds on the price-of-fairness in indivisible item allocation have been proved by \citet{CKKK12,kurz2016price,nicosia2017price,Hale18b,suksompong2019fairly,agnetis2019price,bei2019price,barman2020settling}.
{\citet{argyris2022fair} presented a different approach, which is more related to our approach: they aim to maximize a social welfare measure, subject to the requirement that utility of each agent is at least as high as some reference value.}

\section{UM and EF1}
\label{sec:um-ef1}

It is well-known that EF1 may be incompatible with UM, in the sense that some instances do not admit an allocation that is both UM and EF1.
For example, if Alice's utility for every item is higher than Bob's utility, the only UM allocation gives all items to Alice, which is obviously not EF1 for Bob. Given this incompatibility, we would like to determine whether there exists, among all UM allocations, one that is also fair, but this is computationally challenging.

\renewcommand{\baselinestretch}{1.5} 

\begin{theorem}
\label{um-ef1-stronghard}
The problem  {\sc ExistsUMandEF1} --- deciding whether there exists an allocation that is both UM and EF1 --- is strongly NP-complete.
\end{theorem}
\begin{proof}
The problem is in NP as both UM and EF1 can be tested in polynomial time. 
To prove NP-hardness, we reduce from the following problem, which is known to be strongly NP-hard
\citep{GaJo79a}.
\pbDef{\textsc{3-Partition}}
{
An integer $T>0$, a multiset $\{a_1,\ldots, a_{3m}\}$ of integers with  
$\frac{T}{4}<a_j<\frac{T}{2}$ for all $j\in[3m]$, 
and 
$\sum_{j=1}^{3m} a_j=  mT$.
}
{Can the integers be partitioned into $m$ disjoint triplets such that the sum in each triplet is $T$?}

We construct an instance of 
{\sc ExistsUMandEF1} with $m+1$ agents and $3m+2$ items. The first $3m$ items correspond to the $3m$ integers: their value for the first $m$ agents is equal to the corresponding integer, and their value for agent $m+1$ is $0$.
The last two items are valued at $T$ by the first $m$ agents and $(m/2+1)\cdot T$ by the last agent:
\begin{center}
\scalebox{0.85}{
\begin{tabular}{c|cc|c}
\hline
Items: &$1,\ldots,3m$ &$3m+1,~3m+2$ & Sum \\ 
\hline 
Agents $1,\ldots, m$:  & $v(o_j) = a_j$ & $T$ & $(m+2)T$ \\ 
\hline 
Agent $m+1$: & $0$ & $(m/2+1)\cdot T$ & $(m+2)T$ \\ 
\hline 
\end{tabular} 
}
\end{center} 
{Note that the instance is normalized, that is, the sum of values is the same for all agents.}

In this instance, an allocation is UM if-and-only-if 
the items $3m+1$ and $3m+2$ are given to
agent $m+1$,
and the items $1,\ldots,3m$ are given to agents $1,\ldots,m$.

Suppose we have a ``yes'' instance of \textsc{3-Partition}.	Then, we can allocate the first $3m$ items among the first $m$ agents in a way such that each agent gets utility $T$. We can allocate the items numbered $3m+1$ and $3m+2$ to agent $m+1$. The first $m$ agents are not envious of each other but they are envious of agent $m+1$ whose allocation would give them utility $2T$. However, if one of the items of agent $m+1$ is removed, then envy goes away. Hence there exists a welfare maximizing allocation which is also EF1. 

Now suppose that we have a ``no'' instance of \textsc{3-Partition}. Since there is no equi-partition of the $3m$ elements, there is at least one agent among the first $m$ agents who gets utility less than $T$. This agent envies agent $m+1$ even if one of the two items of agent $m+1$ is removed. Hence there is no UM+EF1 allocation. 
\end{proof}

\begin{remark}
\label{um-efx-stronghard}
{
The problem  {\sc ExistsUMandEFx} is strongly NP-complete: the proof of Theorem \ref{um-ef1-stronghard} holds as-is when EF1 is replaced by EFx, as both items allocated to agent $m+1$ have the same value.}

\end{remark}

As a corollary we get the following hardness results:
\begin{corollary}
\label{um-within-ef1-strongly-hard}
\label{um-within-efx-strongly-hard}
The problems \textsc{ComputeUMwithinEF1}
and \textsc{ComputeUMwithinEFx} are strongly NP-hard.
\end{corollary}
\begin{proof}
We present a polynomial-time one-to-one reduction from {\sc ExistsUMandEF1} (Theorem \ref{um-ef1-stronghard}) to \textsc{ComputeUMwithinEF1}. We use an algorithm for the latter problem and compute the value of the maximum utilitarian welfare within the set of EF1 allocations. Let $w_1$ be this maximum value. Let $w_0$ be the maximum utilitarian welfare without the restriction of being EF1. An allocation maximizing social welfare can be computed in linear time by giving each item to any agent who values it the most.
If $w_0=w_1$, we have a ``yes'' instance of {\sc ExistsUMandEF1}; If $w_0\ne w_1$,  we have a ``no'' instance.
The same proof holds for 
\textsc{ComputeUMwithinEFx}.
\end{proof}

The hardness of  \textsc{ComputeUMwithinEF1}  follows implicitly from Theorem 5 of \citet{BJK+19b}, while not mentioned explicitly there.

Next, we show that (weak) NP-hardness holds even for the case of three agents.

%
%
%
%

%
\begin{theorem}
\label{um-ef1-hard}
The problem {\sc ExistsUMandEF1} for three agents is NP-complete.
\end{theorem}
\begin{proof}
Membership in NP comes directly from Theorem \ref{um-ef1-stronghard}. To prove hardness we reduce from \textsc{Partition}, which is the following problem.

\pbDef{\textsc{Partition}}
{A multiset $\{a_1,\ldots,a_m\}$ of integers, whose sum is $2W$.}
{Is there a partition of the integers into two sets, where the sum in each set is $W$?}

Given an instance of \textsc{Partition}, define an instance of {\sc ExistsUMandEF1} with $m+4$ items: $m$ number-items $\{o_1, \ldots, o_m\}$ and $4$ extra-items $\{e_1, \ldots, e_4\}$. 
There are three agents with the following valuations.
\begin{center}
\scalebox{0.85}{
\begin{tabular}{c|ccccc|c}
\hline
Items: & $o_i$ (for $i\in[m]$) &$e_1$ &$e_2$ &$e_3$ &$e_4$ & sum \\ 
\hline 
Alice:  & $0$ & $W$ & $2W$ & $6W$ & $7W$ & $16W$ \\ 
\hline 
Bob, Chana:& $a_i$ & $3W$& $3W$& $4W$& $4W$& $16W$ \\ 
\hline 
\end{tabular} 
}
\end{center} 
An allocation is UM if-and-only-if
the extra-items $e_3,e_4$ are given to Alice, and the number-items plus $e_1,e_2$ go to Bob or Chana.

Suppose there is an equal partition of the numbers. Then, it is possible to give Bob and Chana a utility of $W$ each from the number-items plus a utility of $3W$ from $e_1,e_2$, for a total of $4W$. The other extra items can be given to Alice. Alice does not envy at all; Bob and Chana do not envy once $e_4$ is removed from Alice's bundle. Hence the allocation is EF1 and UM.

Conversely, suppose there is an EF1 and UM allocation. Bob and Chana value Alice's bundle at $8W$. Once a highest-valued item (for them) is removed from it, they value it at $4W$. Hence, each of them must get a bundle valued at $4W$, so each of them must get one of $\{e_1,e_2\}$ plus a utility of $W$ from the number-items. Hence,  there must be an equal partition of the numbers.
\end{proof}

\begin{remark}
\label{um-efx-hard}
{The problem  {\sc ExistsUMandEFx} is NP-complete for three agents: the proof of Theorem \ref{um-ef1-hard} holds as-is when EF1 is replaced by EFx, as both items allocated to Alice have the same value for Bob and Chana.}
\end{remark}

By arguments similar to  Corollary~\ref{um-within-ef1-strongly-hard},
Theorem~\ref{um-ef1-hard} implies the following hardness result  (which follows from Theorem 4 of \citet{BJK+19b}):

\begin{corollary}
\label{um-within-ef1-hard}
\label{um-within-efx-hard}
The problems \textsc{ComputeUMwithinEF1}  {and \textsc{ComputeUMwithinEFx}}
for three agents is NP-hard.
\end{corollary}

\section{UM and PROP1}
\label{sec:um-prop1}
While EF1 implies PROP1, 
the results for EF1 in Section \ref{sec:um-ef1} do not imply analogous results for PROP1. This is because an algorithm for \textsc{ExistsUMandEF1} might return ``no'' on an instance which admits a UM and PROP1 allocation, and an algorithm for \textsc{ExistsUMandPROP1} might return ``yes'' on an instance which does not admit a UM and EF1 allocation.
Hence, we provide stand-alone proofs for the analogous results for PROP1.

\begin{theorem}
\label{um-prop1-stronghard}
The decision problem \textsc{ExistsUMandPROP1} is strongly NP-complete.
\end{theorem}
\begin{proof}
The problem is in NP as both UM and PROP1 can be tested in polynomial time. To establish NP-hardness, 
we reduce from \textsc{3-Partition} as in Theorem \ref{um-ef1-stronghard}.
We construct an instance with $m+1$ agents and $4m+2$ items and the following valuations:
			 
\begin{center}
\scalebox{0.8}{
\begin{tabular}{c|cc|c}
\hline
Items: &$1,\ldots,3m$ &$3m+1,\ldots,~4m+2$ & Sum \\ 
\hline 
Agents $1,\ldots, m$:  & $v(o_j) = a_j$ & $T$ & $(2m+2) T$\\ 
\hline 
Agent $m+1$: & $0$ & $(1+\frac{m}{m+2})\cdot T$ & $(2m+2) T$ \\ 
\hline 
\end{tabular} 
}
\end{center} 
An allocation is UM if-and-only-if agent $m+1$ gets the items numbered $3m+1$ to $4m+2$.
			 
Suppose we have a ``yes'' instance of \textsc{3-Partition}.	Then, we can allocate the first $3m$ items among the first $m$ agents in a way that each agent gets utility $T$. We can allocate the remaining items to agent $m+1$. 
For the first $m$ agents,
the sum of valuations is $(m+1)\cdot 2T$ so their proportional share is $2 T$.
If they get one of the items numbered from $3m+1$ to $4m+2$, then they get additional utility of $T$ so that their total utility becomes $2T$. Hence, the allocation is PROP1.   

Now suppose that we have a ``no'' instance of \textsc{3-Partition}.  Since there is no equi-partition of the $3m$ elements, at least one agent among the first $m$ agents gets utility less than $T$. This agent does not get utility $2T$ even when adding one of the items numbered from $3m+1$ to $4m+2$. Note that every other item has value less than $T$. Hence there is no UM+PROP1 allocation. 
\end{proof}

\begin{remark}
\label{um-propx-stronghard}
{
The problem  {\sc ExistsUMandPROPx} is strongly NP-complete: the proof of Theorem \ref{um-prop1-hard} holds as-is when PROP1 is replaced by PROPx, as all items allocated to agent $m+1$ have the same value.}

\end{remark}

\begin{corollary}
\label{um-within-prop1-strongly-hard}
\label{um-within-propx-strongly-hard}
The problems 
\textsc{ComputeUMwithinPROP1} 
and
\textsc{ComputeUMwithinPROPx} 
are 
strongly NP-hard.
\end{corollary}

%

Weak NP-hardness persists even for three agents. 
	
\begin{theorem}
\label{um-prop1-hard}
For three agents, 
the decision problem \textsc{ExistsUMandPROP1} 
is NP-complete. 
\end{theorem}
\begin{proof}
The reduction is similar to Theorem \ref{um-ef1-hard}, only with $6$ extra-items and the following valuations:
\begin{center}
\scalebox{0.85}{
\begin{tabular}{c|c|cc|cccc|c}
\hline
& $o_i$ & $e_1$ & $e_2$ &$e_3$ &$e_4$ &$e_5$ &$e_6$ & sum \\ 
\hline 
A: & $0$ & $2W$ & $2W$ & $5W$ & $5W$ & $5W$ & $5W$ & $24W$ \\ 
\hline 
B,C: & $a_i$ & $3W$& $3W$& $4W$& $4W$& $4W$& $4W$& $24W$ \\ 
\hline 
\end{tabular} 
}
\end{center} 
An allocation is UM if-and-only-if Alice gets all and only the items $e_3,\ldots,e_6$. Any such allocation is PROP for Alice.

Bob and Chana value the set of all items at $24 W$, so their proportional share is $8 W$. 
If there is an equal partition of the numbers, then it is possible to give Bob and Chana a bundle worth $4 W$ each, which is proportional after adding to it one of Alice's items.
Conversely, if there is a PROP1 allocation then Bob and Chana's value must be at least $4 W$, so each of them must get at least $W$ from the number-items, so there is an equal partition.
\end{proof}

%

\begin{remark}
\label{um-propx-hard}
{The problem  {\sc ExistsUMandPROPx} is NP-complete for three agents: the proof of Theorem \ref{um-prop1-hard} holds as-is when EF1 is replaced by EFx, as all items allocated to Alice have the same value for Bob and Chana.}
\end{remark}

\begin{corollary}
\label{um-within-prop1-hard}
\label{um-within-propx-hard}
For three agents, the problems \textsc{ComputeUMwithinPROP1} 
{and
\textsc{ComputeUMwithinPROPx}}
are NP-hard. 
\end{corollary}

\section{UM and Fairness for Two Agents}
\label{sec:2-agents}
In this section, we consider the case of two agents. Many fair division problems arise between two parties so it is an important special case to consider. 
	
\begin{theorem}
\label{2agents-ef1-um}
For two agents, there exists a polynomial-time algorithm that solves 
\textsc{ExistsUMandEF1}.
\end{theorem}

\begin{proof}
For each item $o\in O$, we denote Alice's utility by $u(o)$ and Bob's utility by $u(o)+d(o)$. We denote by $O_{eq}$ the set of items for which both agents have the same utility, i.e., $d(o)=0$. We claim that
Algorithm~\ref{alg:um-ef1}
finds a UM and EF1 allocation if-and-only-if such allocation exists.


\begin{algorithm}
	\caption{
	\label{alg:um-ef1}
	Finding a UM and EF1 allocation if one exists; two agents.
	}
	\begin{algorithmic}[1]
	\small
	\STATE Give all items with $d(o)>0$ to Bob and $d(o)<0$ to Alice.
	\FOR {each item $o\in O_{eq}$} \label{line:startloop} 
	\IF {one of the agents is envious}
	\STATE Give $o$ to him/her;
	\ELSE 
	\STATE Give $o$ to an arbitrary agent.
	\ENDIF
	\ENDFOR \label{line:endloop} 
	\IF{the allocation is EF1} \label{line:check-if-fair} 
	\RETURN \label{line:um-ef1-yes} the allocation and say yes
	\ELSE
	\RETURN no
	\ENDIF
		\end{algorithmic}
\end{algorithm}

It is easy to see that an allocation is UM if-and-only-if every item with $d(o)>0$ is given to Bob and every item with $d(0)<0$ is given to Alice. Therefore, line 1 is necessary and sufficient for guaranteeing that the final allocation is UM, regardless of how the remaining items are allocated.

In every UM allocation for two agents, at most one agent is envious --- otherwise the utilitarian welfare could be increased by exchanging the bundles. 
Therefore, throughout the loop in lines \ref{line:startloop}--\ref{line:endloop}, at most one agent is envious. 

We now consider two cases.
First, suppose that the loop in lines \ref{line:startloop}--\ref{line:endloop} gives all items in $O_{eq}$ to a single agent, say Alice. This means that Bob was not envious before the last item was given, so the allocation is EF1 for Bob. 
If the allocation is EF1 for Alice too, then the algorithm says ``yes'' correctly.
Otherwise, \emph{no} allocation is UM and EF1,
as this is the allocation that 
gives Alice the highest possible value among the UM allocations. Then the algorithm says ``no'' correctly.

The second case is that the loop switches from the case ``Alice envies'' to the case ``Bob envies''.
So there is an item $o$ such that, without $o$, Bob does not envy Alice, but once $o$ is given to Alice, Bob envies her. 
At this point, the allocation is EF for Alice and EF1 for Bob. From here, at most one agent envies, and the envy level remains at most one item. Therefore, when the algorithm ends at line \ref{line:um-ef1-yes}, the allocation is UM and EF1.
\end{proof}

Algorithm \ref{alg:um-ef1} 
does not directly solve 
\textsc{ExistsUMandPROP1}
because,
even for 2 agents, EF1 is strictly stronger than PROP1 (see Appendix \ref{sec:ef1-prop1}). 
Therefore, Algorithm \ref{alg:um-ef1} may return ``no'' even though a UM and PROP1 allocation exists.
However, a very similar algorithm can handle UM and PROP1.

\begin{theorem}
\label{2agents-prop1-um}
For two agents, there exists a polynomial-time algorithm that solves 
\textsc{ExistsUMandPROP1}.
\end{theorem}

\begin{proof}
We can use an algorithm almost identical to Algorithm~\ref{alg:um-ef1}, 
except for line \ref{line:check-if-fair} which should read ``if the allocation is PROP1 then''.

Similarly to the previous proof, we consider two cases. 
If the algorithm allocates all items in $O_{eq}$ to a single agent (say Alice), then this is the largest possible value Alice can get in a UM allocation, so a UM-and-PROP1 allocation exists if-and-only-if this final allocation is PROP1.

Otherwise, the allocation switches between ``Alice envies'' and ``Bob envies''. Once this switch happens, the allocation is EF1 and it remains EF1 until the end, so the algorithm finds a UM+EF1 allocation, which is also UM+PROP1.
\end{proof}

\begin{remark}
\label{positive-2agents-eq1-um}
Algorithm~\ref{alg:um-ef1} can be  adapted to solve \textsc{ExistsUMandEQ1}.
Let us change line 3 to ``if one of the agents has a smaller utility than the other agent'', and change line 9 to ``if the allocation is EQ1 then''.

The proof argument is similar to Theorem \ref{2agents-prop1-um}. If the algorithm allocates all items in $O_{eq}$ to a single agent (say Alice), then this is the largest possible value Alice can get in a UM allocation, so a UM-and-EQ1 allocation exists if-and-only-if this final allocation is EQ1.

Otherwise, the allocation switches from ``Alice's utility is smaller'' to ``Bob's utility is smaller''. Once this switch happens, the allocation is EQ1 and it remains EQ1 until the end, so the algorithm finds a UM+EQ1 allocation.
\end{remark}

{
In contrast to the above positive results, we get a  hardness result if we replace EF1 with EFx.
}
\begin{theorem}
\label{2agents-efx-um}
{
The decision problem \textsc{ExistsUMandEFx} 
is NP-complete even for two agents.
}
\end{theorem}
\begin{proof}
By reduction from \textsc{Partition}.
Given an instance of \textsc{Partition} with $m$ integers, define an instance of {\sc ExistsUMandEFx} with $m+2$ items: $m$ number-items $\{o_1, \ldots, o_m\}$ and $2$ extra-items $\{e_1, e_2\}$. 
There are two agents with the following valuations.
\begin{center}
\scalebox{0.85}{
\begin{tabular}{c|ccc|c}
\hline
Items: & $o_i$ (for $i\in[m]$) &$e_1$ &$e_2$ & sum \\ 
\hline 
Alice:  & $10 a_i$ & $2$ & $1$ &  $20W+3$ \\ 
\hline 
Bob   &  $10 a_i$ & $1$& $2$&  $20W+3$ \\ 
\hline 
\end{tabular} 
}
\end{center} 
An allocation is UM if-and-only-if $e_1$ is given to Alice and $e_2$ is given to Bob (regardless of how the number-items are allocated).

If there is an equal partition of the integers, then there exists an allocation in which each agent gets exactly the same value ($10W$) from the number-items. 
So each agent values his/her own bundle at $10W+2$ and the other agent's bundle at $10W+1$. This allocation is EF and therefore EFx.

For the other direction, suppose that there is a UM+EFx allocation of the items. For both agents, the least valuable item in the other agent's bundle has a value of 1. Therefore, each agent must value his/her own bundle at least as much as the other agent's bundle minus 1. Since the values of the number-items are all multiples of 10, both agents must get exactly the same value from these items, so there exists an equal partition of the integers.
\end{proof}

\begin{remark}
\label{2agents-propx-um}
{
The same reduction shows that the problem  \textsc{ExistsUMandPROPx} is NP-complete for two agents.
}
\end{remark}

\begin{corollary}
\label{2agents-um-within-efx}
\label{2agents-um-within-propx}
{
The optimization problems \textsc{ComputeUMwithinEFx} 
and
\textsc{ComputeUMwithinPROPx} 
are NP-hard even for two agents.
}
\end{corollary}

For EF1 and PROP1 we showed that, for three agents, both the decision problems (UM and fair) and the corresponding maximization problems (UM within fair) are NP-hard.
Below we show that, for two agents, there is a substantial gap between the decision and the maximization problems: while deciding the existence of UM and fair allocations is in P, computing an allocation that is UM within the fair allocations is NP-hard.
\citet{BJK+19a,BJK+19b} proved this hardness for EF1.%
\footnote{
{
In fact, \citet{BJK+19b}[Appendix A] prove a stronger inapproximability result: it is NP-hard to obtain an $m^{\delta}$ factor approximation to the maximal utilitarian welfare subject to EF1, where $\delta>0$ is a fixed constant. Note that they use the term \textsc{FA-EF1} for the problem we call \textsc{ComputeUMwithinEF1}.
}
}
\begin{theorem}[\citet{BJK+19a}]\label{th:um-within-nph}
 For two agents, the problem
 \textsc{ComputeUMwithinEF1}
 is NP-hard.
\end{theorem}

Below we show, using a different reduction, that the hardness holds for the weaker PROP1 condition too.

\begin{theorem}\label{th:um-within-prop1}
For two agents, the problem
\textsc{ComputeUMwithinPROP1}
is NP-hard.
\end{theorem}

\begin{proof}
The proof is by reduction from \textsc{Knapsack}.


\pbDef{\textsc{Knapsack}}
{A set $M$ of $m$ elements, and a threshold value $T$. Each item $i$ has value $v_i$ and weight $w_i$. }
{
For each subset $S$ of elements,
denote its value by $v(S):= \sum_{i \in S} v_i$ and its weight by $w(S)  := \sum_{i \in S} w_i$.
What is a set $S\subset M$ that maximizes $v(S)$ subject to $w(S)\leq T$?}

Let 
$W := \sum_{i\in M} w_i$ be the sum of all weights, 
$w^*=\max_{i\in M}w_i$ be the largest weight
and
$V := \sum_{i\in M} v_i$ be the sum of all values.

Initially, we assume that $T\ge W/2$.
We construct a fair allocation instance $m+2$ items: $m$ ``usual items'' $o_i$ and two ``big items'' $o_A,o_B$, 
and two agents with the following valuations:

\begin{center}
\scalebox{0.85}{
\begin{tabular}{c|ccc|c}
\hline
Items: & $o_i$ (for $i\in[m]$) &$o_A$ &$o_B$ & sum \\ 
\hline 
Alice:  & $w_i$ & $2T-W+w^*$ & $w^* $ &  $2T+2w^*$ \\ 
\hline 
Bob:  &  $w_i + v_i$ 
&
$2T+V+w^*$
&
$W+V+w^*$
&
$2W+3V+2T+2w^*$ \\ 
\hline 
\end{tabular} 
}
\end{center}

%

We prove that
maximizing the utilitarian welfare subject to PROP1 in the allocation instance
is equivalent to 
maximizing the value subject to the weight constraint in the \textsc{Knapsack} instance.
We make several observations.

\textbf{Observation 1.}  Bob's value for every item is larger than Alice's value. Therefore, the unique UM allocation gives all items to Bob. 
Denote its utilitarian welfare by $U_{\max}$.
This allocation is not PROP1; to construct a PROP1 allocation, we must move some items from Bob to Alice; to construct an allocation that is UM within PROP1, we must choose which items to move such that the utilitarian welfare does not decrease too much w.r.t. $U_{\max}$.

\textbf{Observation 2.}  For any big item, Alice's value is smaller than Bob's value by $W+V$. So 
in any allocation in which a big item is allocated to Alice, the utilitarian welfare is at most $U_{\max}-W-V$.

\textbf{Observation 3.}  
Denote the allocation in which Alice gets all number-items and Bob gets all big items by $p_0$.
Its utilitarian welfare is 
$U_{\max}-V$. It is EF1 (and PROP1) since:
\begin{itemize}
\item $u_A(p_0(A)) = W$, while Alice's valuation to Bob's bundle without $o_A$ is only $w^*\leq W$;
\item $u_B(p_0(B)) = W+2V+2T+2w^*$, while Bob's valuation to Alice's bundle is only $W+V$.
\end{itemize}
Any allocation that is UM within PROP1 must have a utilitarian welfare at least as high as $p_0$, that is, at least $U_{\max}-V > U_{\max}-V-W$.
By observation 2, any such allocation must 
give both big items to Bob.
We focus only on these allocations from now on; let us call these allocations ``reasonable''.

There is a one-to-one correspondence between the subsets of knapsack items (subsets of $M$) to the set of reasonable allocations.
Each subset $S\subset M$ corresponds to a reasonable allocation, $p_S$, in which: (1) The big items, and all usual items in $S$, are given to Bob; and (2) all usual items not in $S$ are given to Alice.
In the allocation $p_S$,
the utilitarian welfare coming from the big items is $2T+2 V + W+2w^*$. 
The utilitarian welfare coming from the usual items is $W + v(S)$.
Therefore, the total welfare in $p_S$ is $C+v(S)$, where $C := 2W+2 V + 2T+2w^*$. Note that $C$ is a constant that does not depend on the selection of $S$.
Thus, finding a reasonable allocation $p_S$ with maximum welfare is equivalent to finding a subset $S\subseteq M$ with maximum value.

The value of $o_A$ and $o_B$ for Bob is so big, that Bob never envies Alice regardless of how the usual items are allocated. 
Moreover, for Alice, $o_A$ is the most valuable item in Bob's bundle (since $T\geq W/2$). Therefore, an allocation is PROP1 if-and-only-if, when $o_A$ is added to Alice's bundle,  she will have at least half of her total value.
Alice's valuation of $o_A$ is $2T-W+w^*$, and her valuation of her own bundle is $W-w(S)$. Therefore, the PROP1 condition for Alice is:

\begin{align*}
W-w(S) + &\left(2T-W+w^*\right) ~\geq~ T+w^*
\end{align*}
which holds if-and-only-if $w(S)\leq T$. 
Therefore, maximizing utilitarian welfare subject to PROP1 is equivalent to maximizing $v(S)$ subject to $w(S)\leq T$, as claimed.

In the complementary case $T < W/2$, the reduction is similar, but with one more big item $o_C$:

\begin{center}
\scalebox{0.85}{
\begin{tabular}{c|cccc|c}
\hline
Items: & $o_i$ (for $i\in[m]$) &$o_A$ &$o_B$ &$o_C$ & sum \\ 
\hline 
Alice:  & $w_i$ & $W-2T+2w^*$ & $W-2T+w^*$ & $W-2T+w^*$ & $4W-6T+4w^*$ \\ 
\hline 
Bob:  &  $w_i + v_i$ 
&
$W-2T+2w^*+V$
&
$W-2T+w^*+V$
&
$W-2T+w^*+V$ 
&
$4W-6T+4w^*+4V$ 
\\ 
\hline 
\end{tabular} 
}
\end{center}

Now, all big items $o_A$, $o_B$ and $o_C$ are allocated to Bob;
$o_A$ is still most valuable to Alice.
Alice's valuation of item $o_A$ is $W-2T+2w^*$, and her valuation of her own bundle is $W-w(S)$. Therefore, the PROP1 condition for Alice is:

\begin{align*}
W - w(S) +\left(W-2T+2w^*\right) ~\geq~ 2W-3T+2w^*
\end{align*}
which
holds if-and-only-if $w(S)\leq T$ as before.
\end{proof}

The hardness proofs in both 
Theorem \ref{th:um-within-prop1} and \citet{BJK+19b} use an instance that is not normalized --- the sum of Alice's valuations is smaller than the sum of Bob's valuations. 
In many situations, the agents' valuations are normalized
such that the sum of valuations is the same for all agents.
For example, in the popular fair division website \url{spliddit.org}
each person enters utilities that sum up to $1000$. 
Therefore, a natural question is whether the two problems
\textsc{ComputeUMwithinEF1}
and \textsc{ComputeUMwithinPROP1}
remain NP-hard with two agents even when restricted to normalized valuations.
Despite many efforts, we could not modify the reduction of Theorem \ref{th:um-within-prop1} to work with normalized valuations. 
\begin{open}
\label{open:normalized}
For two agents with normalized valuations, 
what is the run-time complexity of computing an allocation that is UM within EF1, or UM within PROP1?
\end{open}

\section{{UM within Fairness for Few Agents}}
\label{sec:fixedn}
We have proved that all our problems are strongly NP-hard when the number of agents is unbounded,
and weakly NP-hard when the number of agents is fixed. 
This raises the question of whether the problems are strongly NP-hard when the number of agents is fixed. 
We show that the answer is ``no'' by presenting pseudopolynomial time algorithms for  \textsc{ComputeUMwithinEF1} and \textsc{ComputeUMwithinPROP1}, and hence for deciding \textsc{ExistsUMandEF1} and \textsc{ExistsUMandPROP1} (see Corollary \ref{um-within-ef1-strongly-hard}).

Our algorithms assume that valuations are non-negative integers, and the sum of all values for a single agent is upper-bounded by some integer $V$. The run-time is polynomial in $V$. Therefore, in the special case in which the valuations are \emph{binary}, i.e., $u_i(o)\in\{0,1\}$ for all $i\in N, o\in O$, we have $V\leq m$, so the run-time is in $O(\text{poly}(m))$.
This special case has been studied substantially in the fair division literature \citep{bouveret2016characterizing,barman2018groupwise,halpern2020fair,darmann2015maximizing,barman2018greedy}.

The run-time of our algorithms is exponential in $n$. This is not surprising in view of the hardness results for unbounded $n$.
Our goal in the present paper is to prove that the problems are not strongly NP-hard when $n$ is fixed; we leave to future work
the problem of {finding algorithms with a smaller exponent}.

Our algorithms use dynamic programming. Each algorithm keeps a set of \emph{states}, which initially contains a single state representing the empty allocation. 
There are $m$ iterations. 
At each iteration $k$, the algorithm considers, for every state in the current set, all $n$ possible ways to allocate item $k$. For each option, a new state is added (if it does not exist yet).
Finally, the algorithm picks the optimal allocation from among the states generated at the last (the $m$-th) iteration. The run-time of such an algorithm is upper-bounded by the total number of possible states.
To flesh out this scheme, we have to specify, for each algorithm, what the states are, and how the next states are computed at each iteration.

As a warm-up, we present an algorithm for \textsc{ComputeUMwithinPROP}. Such an algorithm may be useful when a proportional allocation exists.
\begin{theorem}
\label{um-within-prop-pseudopoly}
Given $n$ agents and $m$ items, 
when all valuations are positive integers and the sum of all values for a single agent is at most $V$, 
it is possible to compute, in time $O(m V^n)$, a UM within PROP allocation (whenever a PROP allocation exists), or detect that a PROP allocation does not exist.
\end{theorem}

\begin{proof}
The states are of the form $(k; t_1,\ldots,t_n)$, where $k\in\{0,\ldots,m\}$ and $t_i \in[0,V]$ for all $i\in[n]$. Each such state represents the fact that, by allocating items $o_1,\ldots,o_k$, it is possible to give a utility of $t_i$ to each agent $i$. 
The initial state is $(0; 0,\ldots,0)$, corresponding to the empty allocation.

Each state $(k-1; t_1,\ldots,t_n)$ with $k\in[m]$ has $n$ next states: for all $a\in[n]$, a next state $(k; t_1,\ldots,t_a+u_a(o_{k}),\ldots,t_n)$ corresponds to allocating the next item $o_{k}$ to agent $a$.

The states $(m; t_1,\ldots,t_n)$ correspond to final allocations. A state corresponds to a proportional allocation if and only if $t_i\geq v_i(O)/n$ for all $i\in[n]$. 
If there are no such states, then we return ``No proportional allocation exists''.
Otherwise, we choose a state in which the sum $t_1+\ldots+t_n$ is maximum; this sum represents the largest utilitarian welfare compatible with proportionality.
The maximizing allocation can be found by backtracking the construction of the states.

The total number of possible states in each iteration is $V^n$, so the total number of possible states overall is $O(m V^n)$.
\end{proof}

The above scheme cannot be directly applied to \textsc{ComputeUMwithinPROP1}. This is because PROP1 does not give a unique value-threshold for every agent: the value-threshold depends on the highest-valued item that is not assigned to that agent.
To handle this we need more elaborate states.

\begin{theorem}
\label{um-within-prop1-pseudopoly}
Given $n$ agents and $m$ items, 
when all valuations are positive integers and the sum of all values for a single agent is at most $V$, 
it is possible to compute a UM within PROP1 allocation in time 
{$O(m^{n+1} V^n)$}.
\end{theorem}

\begin{proof}
The states are of the form $(k; t_1,\ldots,t_n,b_1,\ldots,b_n)$, where $k\in\{0,\ldots,m\}$ and $t_i \in[0,V]$ and $b_i\in O\cup \{\emptyset\}$ for all $i\in[n]$. 
Each such state represents the fact that there is an allocation of items $o_1,\ldots,o_k$, where each agent $i$ gets a utility of $t_i$, and the most valuable item given to agents other than $i$ is $b_i$.
The motivation for recording these items in the states is that, in order to check the PROP1 condition, we have to add to each agent's value, the value of an item that is allocated to another agent.

The initial state is $(0; 0,\ldots,0; \emptyset,\ldots,\emptyset)$. It represents the empty allocation, where each agent gets a utility of $0$ and no items.

For each state $(k-1; t_1,\ldots,t_n,b_1,\ldots,b_n)$ with $k\in[m]$, there are $n$ next states, constructed as follows. 
For all $a\in[n]$, there is a next state 
$(k;
t_1,\ldots,t_a+u_a(o_{k}),\ldots,t_n; b'_1,\ldots,b'_n)$, where
\begin{align}
\label{eq:new_largest_item_prop1}
b'_{i} := 
\begin{cases}
o_k & \text{~if~} i\neq a \text{~and~} u_i(o_k)> u_i(b_{i}) \text{~~(new most-valuable item for $i$ in $p(a)$);}
\\
b_{i} & \text{otherwise.}
\end{cases}
\end{align}


The states $(m; t_1,\ldots,t_n,b_1,\ldots,b_n)$ correspond to final allocations. A state corresponds to a PROP1 allocation if and only if $t_i + v_i(b_i)\geq v_i(O)/n$ for all $i\in[n]$. 
Note that there must be at least one such state, since a PROP1 allocation always exists.
From these states, we choose one that maximizes the sum $t_1+\ldots+t_n$; this sum represents the largest utilitarian welfare that is compatible with PROP1.
The maximizing allocation can be found by backtracking the construction of the states.

{The total number of possible states in each iteration is $V^n\cdot m^n$, so the total number of possible states overall is $O(m^{n+1}V^n)$.}
\end{proof}

\begin{remark}
\label{um-within-propx-pseudopoly}
{
With slight modifications in the state update rules, we can compute a UM-within-fair allocation for various other fairness criteria:
}

{For  UM-within-PROPx \citep{li2021almost}, interpret the $b_i$ in the state as: the \emph{least} valuable item given to agents other than $i$, and replace \eqref{eq:new_largest_item_prop1} by
\begin{align}
b'_{i} := 
\begin{cases}
o_k & \text{~if~} 
i \neq a \text{~and~} \underline{u_i(o_k)< u_i(b_{i}) \text{~or~} b_i=\emptyset} \text{~~(new least-valuable item for $i$ in $p(a)$);}
\\
b_{i} & \text{otherwise.}
\end{cases}
\end{align}
For UM-within-EQ1, interpret $b_i$ as the most valuatble item given to \emph{agent $i$},
and let
\begin{align}
b'_{i} := 
\begin{cases}
o_k & \text{~if~} 
\underline{i = a} \text{~and~} u_i(o_k)> u_i(b_{i}) \text{~~(new most-valuable item for $i$ in $\underline{p(i)}$);}
\\
b_{i} & \text{otherwise.}
\end{cases}
\end{align}
A state corresponds to an EQ1 allocation if and only if $t_j \geq 
t_i - v_i(b_i)$ for all $i,j\in[n]$. 
For  UM-within-EQx \citep{freeman2019equitable}, interpret $b_i$ as the least valuatble item given to agent $i$, and let
\begin{align}
b'_{i} := 
\begin{cases}
o_k & \text{~if~} 
i = a \text{~and~} \underline{u_i(o_k)< u_i(b_{i}) \text{~or~} b_i=\emptyset} \text{~~(new least-valuable item for $i$ in $p(i)$);}
\\
b_{i} & \text{otherwise.}
\end{cases}
\end{align}
}
\end{remark}

Next, we give a pseudopolynomial time algorithm for UM-within-EF.

\begin{theorem}
\label{um-within-ef-pseudopoly}
Given $n$ agents and $m$ items, 
when all valuations are positive integers and the sum of values for a single agent is at most $V$, 
it is possible to compute, in time $O(m V^{n(n-1)})$, a UM within EF allocation (whenever an EF allocation exists), or detect that an EF allocation does not exist.
\end{theorem}

\begin{proof}
The states are of the form $(k; (t_{i,j})_{i\neq j})$, where $k\in\{0,\ldots,m\}$ and $t_{i,j} \in[-V,V]$ for all $i,j\in[n]$ such that $i\neq j$.
Each such state represents the fact that there exists an allocation of items $o_1,\ldots,o_k$, 
in which, for all $i,j\in[n]: u_i(p(i)) - u_i(p(j)) = t_{i,j}$.
The initial state is $(0; 0,\ldots,0)$, corresponding to the empty allocation.

Each state $(k-1; (t_{i,j})_{i\neq j})$ with $k\in[m]$ has $n$ next states: for all $a\in[n]$, 
there is a next state $(k; (t'_{i,j})_{i\neq j})$
that corresponds to allocating the next item $o_{k}$ to agent $a$, where:
\begin{align}
\label{eq:new_diff}
t'_{i,j} = 
\begin{cases}
t_{i,j} + u_i(o_k) & i=a \text{~~(adding $o_k$ to the difference for the receiver);}
\\
t_{i,j} - u_j(o_k) & j=a\text{~~(subtracting $o_k$ from the difference for the non-receivers);}
\\
t_{i,j} & \text{otherwise.}
\end{cases}
\end{align}

The states $(m; (t_{i,j})_{i\neq j})$ correspond to final allocations.
A state corresponds to an envy-free allocation if and only if $t_{i,j}\geq 0$ for all $i\neq j$. 
If there are no such states, then we return ``No envy-free allocation exists''.
Otherwise, we choose a state that maximizes the sum of all $n(n-1)$ elements, $t_{1,2}+t_{1,3}+\ldots+t_{n,n-1}$. We claim that this corresponds to maximizing the utilitarian welfare. To see this, note that for each agent $i\in[n]$:
\begin{align*}
\sum_{j\neq i, j=1}^n t_{i,j} 
&= \sum_{j\neq i, j=1}^n [u_i(p(i)) - u_i(p(j))] 
\\
&= \sum_{j=1}^n [u_i(p(i)) - u_i(p(i))]
&& \text{~(adding $[u_i(p(i)) - u_i(p(i))]$)~}
\\
&= \left(\sum_{j=1}^n u_i(p(i))\right) - 
\left(\sum_{j=1}^n u_i(p(j))
\right)
\\
&=
n\cdot u_i(p(i)) - u_i(O).
\end{align*}
As the $u_i(O)$ are constants that do not depend on the allocation, a higher sum $t_{1,2}+t_{1,3}+\ldots+t_{n,n-1}$  corresponds to a higher sum $u_1(p(1))+\cdots +u_n(p(n))$, which is the utilitarian welfare. The maximizing allocation can be constructed by backtracking the construction of states.

The total number of possible states in each iteration is $V^{n(n-1)}$, so the total number of possible states overall is $O(m V^{n(n-1)})$.
\end{proof}

Finally, to compute UM within EF1, we need more state elements tracing the most valuable objects given to agents. 

\begin{theorem}
\label{um-within-ef1-pseudopoly}
Given $n$ agents and $m$ items, 
when all valuations are positive integers and the sum of all values for a single agent is at most $V$, 
it is possible to compute a UM within EF1 allocation in time {$O(m^{n(n-1)+1}\cdot V^{n(n-1)}) \approx O(m^{n^2} V^{n^2})$}.
\end{theorem}

\begin{proof}
The states are of the form $(k; (t_{i,j})_{i\neq j};
(b_{i,j})_{i\neq j}
)$, where $k\in\{0,\ldots,m\}$ and $t_{i,j} \in[-V,V]$ and 
$b_{i,j}\in O \cup \{\emptyset\}$ 
for all $i,j\in[n]$ such that $i\neq j$.
Each such state represents the fact that there is an allocation of items $o_1,\ldots,o_k$
in which, for all $i,j\in[n]: u_i(p(i)) - u_i(p(j)) = t_{i,j}$,
and in addition, item $b_{i,j}$ is the item that maximizes $u_i$ in $p(j)$.
We allow $b_{i,j}$ to be $\emptyset$, to handle the case in which $p(j)$ is empty.

The initial state is $(0; 0,\ldots,0; \emptyset,\ldots,\emptyset)$, which represents the empty allocation.
For each state $(k-1; (t_{i,j})_{i\neq j};
(b_{i,j})_{i\neq j}
)$ with $k\in[m]$, 
there are $n$ next states.
For each agent $a\in [n]$, 
the next state is
$(k; (t'_{i,j})_{i\neq j};
(b'_{i,j})_{i\neq j})$, where 
$t'_{i,j}$ is defined as in \eqref{eq:new_diff},
and 
\begin{align}
\label{eq:new_largest_item}
b'_{i,j} = 
\begin{cases}
o_k & j=a \text{~and~} u_i(o_k)> u_i(b_{i,j}) \text{~~(new most-valuable item for $i$ in $p(a)$);}
\\
b_{i,j} & \text{otherwise.}
\end{cases}
\end{align}

As in the proof of Theorem \ref{um-within-ef-pseudopoly},  maximizing utilitarian welfare is equivalent to maximizing the sum of $t_{1,2}+\dots+t_{n,n-1}$. So we choose a state maximizing this sum such that for all $i,j$ we have $t_{i,j}+b_{i,j}\ge 0$. This corresponds to an EF1 allocation maximizing the utilitarian welfare, which can be constructed by backtracking.

The total number of possible states in each iteration is $m^{n(n-1)}\cdot V^{n(n-1)}$, so the total number of possible states overall is $O(m^{n(n-1)+1}\cdot V^{n(n-1)})$.
\end{proof}

\begin{remark}
\label{um-within-efx-pseudopoly}
{
We can adapt the algorithm to find an allocation that is UM-within-EFx \citep{caragiannis2019envy}.
The interpretation of the elements $b_{i,j}$ in the states should change to : the item that \emph{minimizes} $u_i$ in $p(j)$. Equation \eqref{eq:new_largest_item} should be modified to:
\begin{align}
b'_{i,j} = 
\begin{cases}
o_k & j=a \text{~and~} (u_i(o_k)< u_i(b_{i,j}) \text{~or~} b_{i,j}=\emptyset)
\\
b_{i,j} & \text{otherwise.}
\end{cases}
\end{align}
}
\end{remark}

\begin{remark}
We need $n(n-1)$ variables $b_{i,j}$ because the most-valuable item in $p(j')$ can be different for different agents $j$. 
In the special case in which the instance is \emph{ordered} (all agents rank the items in the same order), we can use a single  $b_{j}$ for each agent, and the run-time becomes $O(m^{n+1}V^{n(n-1)})$.
\end{remark}

\section{{Empirical Evaluation}}
Allocating indivisible items fairly is receiving growing attention in many application areas including allocating time slots \cite{GoPr14a}, recommendations online \cite{Burke:AlgoFairness,chakraborty2019equality}, rides in taxies \cite{dickerson2018allocation}, and conference papers for review \cite{LMNW18a,stelmakh2021peerreview4all}. 
To check the practical applicability of computing efficient and fair allocation, in this section we investigate the runtime of the dynamic programmig algorithms of Section \ref{sec:fixedn} for UM within EF/EF1 and UM within PROP/PROP1 compared with direct implementations as Mixed Integer Linear Programs (MILPs). To this end, we implemented a simple allocation ILP in Gurobi 9.1 using the Python interface \footnote{Source Code: \url{https://github.com/tu-dai/IndivisibleAllocation}} and compared it with our direct dynamic programming implementations.\footnote{Source Code: \url{https://github.com/erelsgl/dynprog}}

\subsection{Mixed Integer Linear Programs for Finding Utilitarian Maximal Allocations with EF/EF1 and PROP/PROP1}

In this section we give formulations for the mixed integer linear programs that we used in our experiments. We begin with UM within Prop. Given a binary variable $assigned[a,o]$ for each $a \in N$ and $o \in O$, which takes value $1$ if agent $a$ is assigned object $o$ in the matching, as well as a continuos variable $utility[a]$ for each agent $a \in N$ which will take the total utility for agent $a$ for the assignment, we can find the utilitarian maximal (UM) assignment within the set of Proportional (PROP) assignments as follows.

\begin{center}
$
\begin{array}{|ll|l|}%
\hline
\max & \sum_{a\in N} utility[a], s.t., & \\
\text{(1)} & utility[a] = \sum_{o \in O} assigned[a,o] \cdot u_a[o] & \forall a \in N \\
\text{(2)} & utility[a] \geq \nicefrac{u_a[O]}{|N|}  & \forall a \in N \\
\hline
\end{array}
$
\end{center}

\noindent
We can use a similar formulation for the UMinEF formulation. Note that the constraint (2) needs to be for all $n(n-1)$ ordered pairs of $i,j \in N$, as envy is not necessarily a symmetric relation.

\begin{center}
$
\begin{array}{|ll|l|}%
\hline
\max & \sum_{a\in N} utility[a], s.t., & \\
\text{(1)} & utility[a] = \sum_{o \in O} assigned[a,o] \cdot u_a[o] & \forall a \in N \\
\text{(2)} & utility[i] \geq \sum_{o \in O} assigned[j,o] \cdot u_i[o] & \forall i,j \in N \\
\hline
\end{array}
$
\end{center}

The EF1/PROP1 formulations are slightly more complicated to encode without too large an increase in model size. Informally, for the PROP1 setting we add an auxiliary variable to track the highest value item (to agent $i$) that has not been allocated to agent $i$. We use a similar trick but for every pair $i,j$ of agents for EF1. Specifically, for PROP1, given additional continuos variable $max\_not\_assigned[a]$ for each $a \in N$ that tracks the value of the maximal unassigned object and binary variable $not\_assigned[a,o]$ for every $a \in N$ and $o \in O$ to track what objects are unassigned to agent $a$, we can express the UMinPROP1 MILP as follows.

\begin{center}
$
\begin{array}{|ll|l|}%
\hline
\max & \sum_{a\in N} utility[a], s.t., & \\
\text{(1)} & utility[a] = \sum_{o \in O} assigned[a,o] \cdot u_a[o] & \forall a \in N \\
\text{(2)} & not\_assigned\_best[a,o] \leq 1 - assigned[a,o] & \forall a \in N, o \in O \\
\text{(3)} & \sum_{o \in O} not\_assigned\_best[a,o] \leq 1 & \forall a \in N \\
\text{(4)} & value\_not\_assigned\_best[a] \leq  \sum_{o \in O} not\_assigned\_best[a,o] \cdot u_a[o] & \forall a \in N \\
\text{(5)} & utility[a] + value\_not\_assigned\_best[a] \geq \nicefrac{u_a[O]}{|N|}  & \forall a \in N \\
\hline
\end{array}
$
\end{center}
Constraints (2) and (3) enforce that at most one item $o$, which is not assigned to $a$, has $not\_assigned\_best[a,o]=1$; constraint (4) sets $value\_not\_assigned\_best[a]$ to at most the largest value of a not-assigned item;
constraint (5) uses this value to enforce the PROP1 constraint.

Finally, for UMinEF1 we use a similar idea as our UMinPROP1 formulation but instead of $not\_assigned\_best$ for each agent, we instead introduce variable $not\_assigned\_best[i,j,o]$ for each ordered pair $i,j \in N$ and object $o \in O$ that captures the best, according to $i$, item assigned to $j$ that is not assigned to $i$, and we also must track $value\_not\_assigned\_best[i,j]$ for each ordered pair $i,j \in N$.

\begin{center}
$
\begin{array}{|ll|l|}%
\hline
\max & \sum_{a\in N} utility[a], s.t., & \\
\text{(1)} & utility[a] = \sum_{o \in O} assigned[a,o] \cdot u_a[o] & \forall a \in N \\
\text{(2)} & not\_assigned\_best[i,j,o] \leq 1-assigned[i,o] & \forall i,j \in N, o \in O \\
\text{(3)} & not\_assigned\_best[i,j,o] \leq assigned[j,o] & \forall i,j \in N, o \in O \\
\text{(4)} & \sum_{o \in O} not\_assigned\_best[i,j,o] \leq 1 & \forall i,j \in N \\
\text{(5)} & value\_not\_assigned\_best[i,j] \leq  \sum_{o \in O} not\_assigned\_best[i,j,o] \cdot u_i[o] & \forall a \in N \\
\text{(6)} & utility[i] + value\_not\_assigned\_best[i,j]\geq \sum_{o \in O} assigned[j,o] \cdot u_i[o]  & \forall i \in N \\
\hline
\end{array}
$
\end{center}

\subsection{Experimental Details and Results}

We generate synthetic data for our experiments using a Mallows model and assigning agents Borda utilities, i.e., $m-1, m-2, \ldots 0$ for the $m$ items as is often done in the empirical literature in this area \cite{MaWa13a,MaWa17}. A Mallows model is controlled by a \emph{dispersion parameter}, $\phi$, which changes the distribution of preferences around a reference ranking \cite{mallows1957non}. Informally, Mallows models allow us to simulate the situation when all agents have identical preferences, $\phi = 0.0$, and the situation when agents have preferences drawn uniformly at random, $\phi = 1.0$, and every point in between as parameterized by the Kendall-Tau (sometimes called the swap) distance between the ordinal rankings. Hence we are able to test what the impact of agents having correlated preferences has on the runtime of the algorithms.

\begin{figure}
\centering
\includegraphics[width=\linewidth]{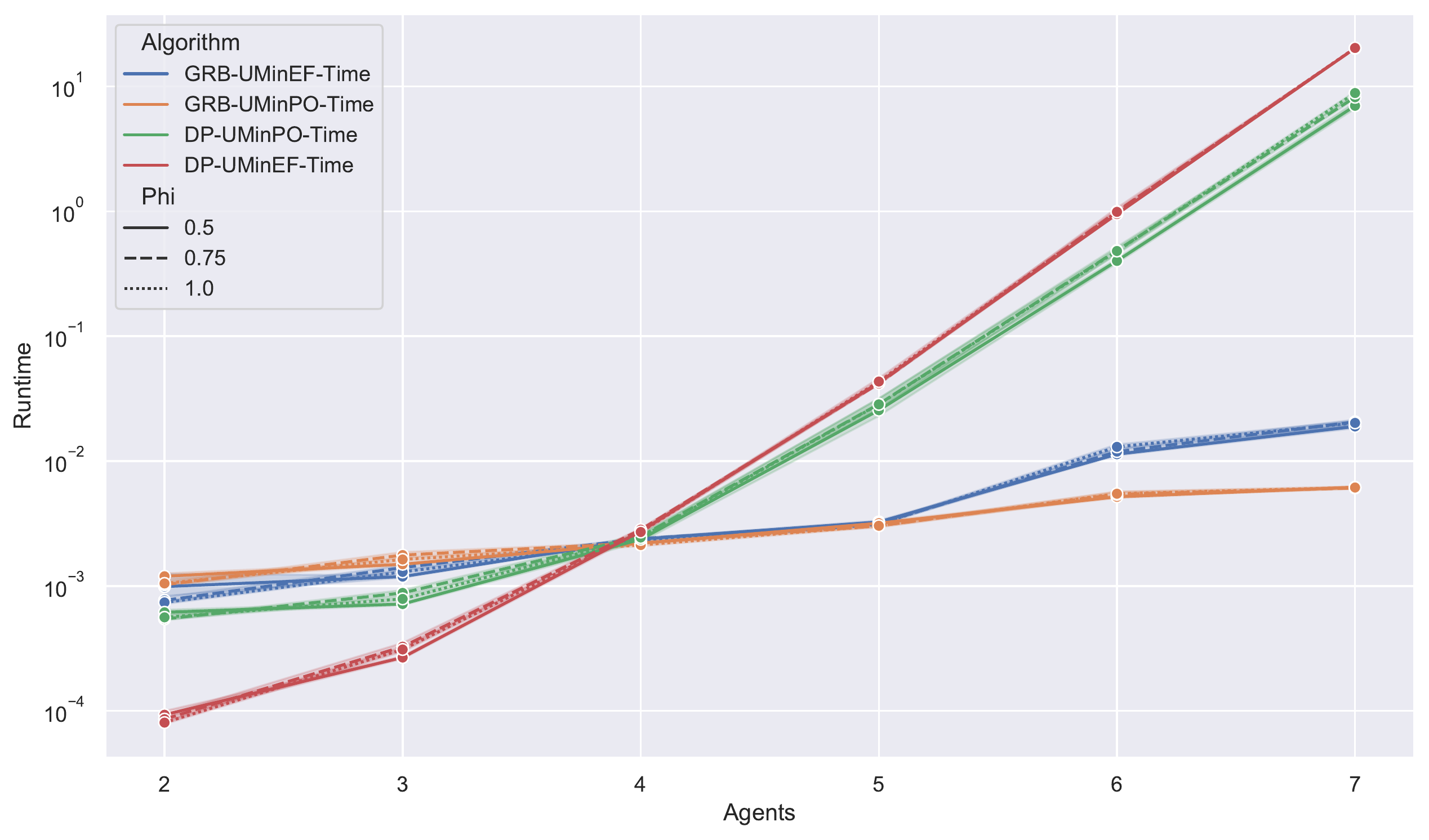}
\caption{Run-time (on a logarithmic scale) of the Dynamic Programming algorithms (DP) and Mixed Integer Linear Programming (GRB) versions of algorithms to find UM within EF or PROP allocations. We plotted the run-times for three different settings of the Mallows model dispersion parameter $\phi$; the plots for these three settings are almost completely overlapping, though if we do not display on a log scale y-axis, we can see a small difference in runtime. The MILP implementations are able to scale much better as we increase the number of agents and objects. At $n=7$, the bottom plot is GRB-PROP (denoting the fastest performance), the next one is GRB-EF, then DP-PROP, and finally DP-EF.
}
\label{fig:EF-PO}
\end{figure}

For all our settings we held the number of agents and items to be the same, $m=n$, and swept this value between 2 and 7. For each step we generated 50 instances as described above for $\phi \in \{0.5, 0.75. 1.0\}$. All experiments were run on a 2018 Mac Book Pro with 2.6 GHz 6-Core Intel Core i7 and 32 GB of RAM. Our results (with a log scale y-axis) are depicted in Figures \ref{fig:EF-PO} and \ref{fig:EF1-PO1}.

\begin{figure}
\centering
\includegraphics[width=\linewidth]{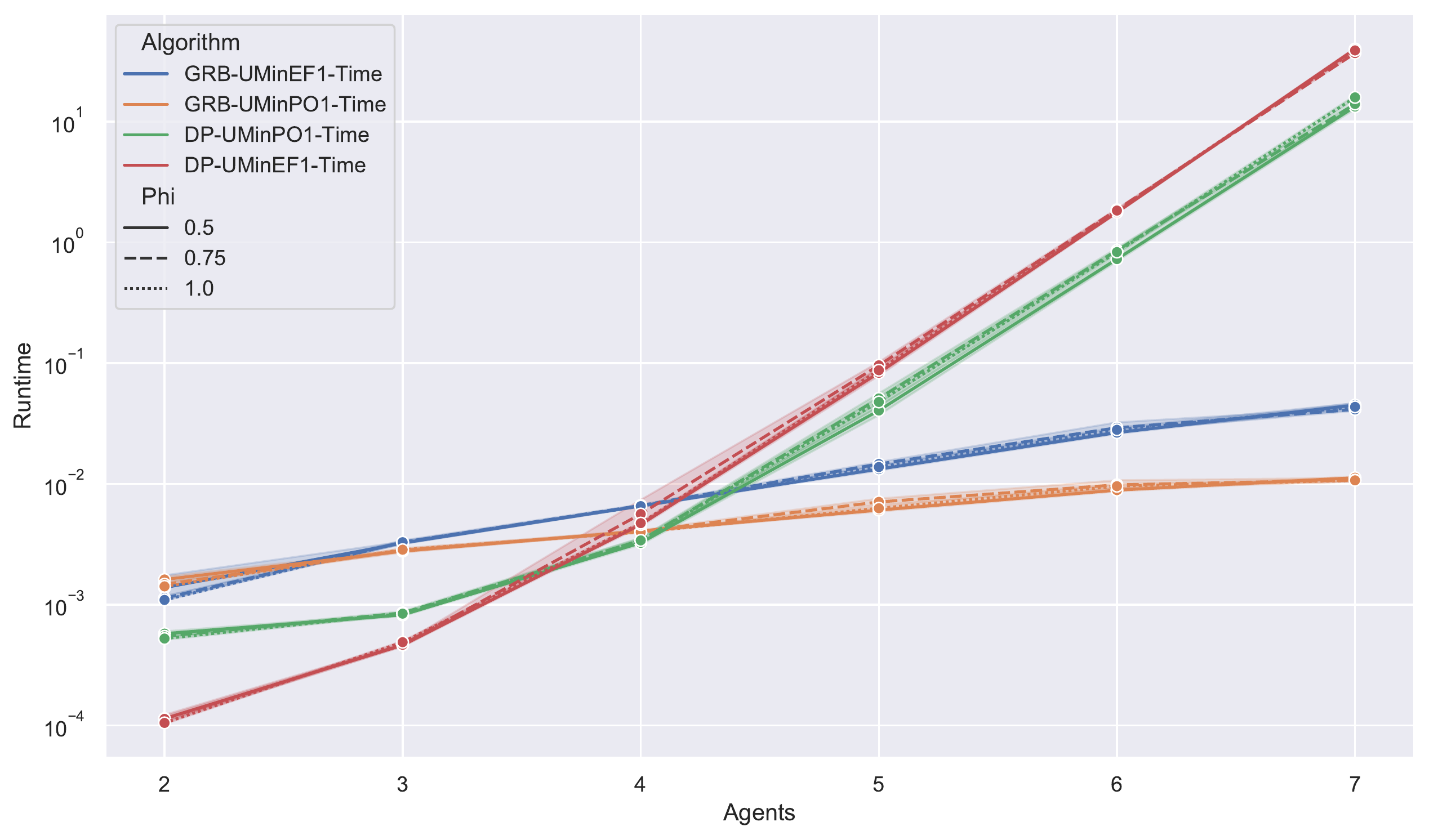}
\caption{Run-time (on a logarithmic scale) of the Dynamic Programming algorithms (DP) and Mixed Integer Linear Programming (GRB) versions of algorithms to find UM within EF1 or PROP1 allocations. We plotted the run-times for three different settings of the Mallows model dispersion parameter $\phi$; the plots for these three settings are almost completely overlapping, though if we do not display on a log scale y-axis, we can see a small difference in runtime. The MILP implementations are able to scale much better as we increase the number of agents and objects. At $n=7$, the bottom plot is GRB-PROP1 (denoting the fastest performance), the next one is GRB-EF1, then DP-PROP1, and finally DP-EF1.
}
\label{fig:EF1-PO1}
\end{figure}

Looking first at Figure \ref{fig:EF-PO}, the UM within EF versions of both the DP and the MILP versions of the algorithms take more time to solve than the UM within PROP. the MILP implementations are able to scale much better as we increase the size of the instances. The runtime over the 50 samples is fairly constant for each of the algorithms. While it is hard to distinguish in the graph, there is a small difference in average runtime as we change $\phi$, generally speaking, less correlated preferences, i.e., lower $\phi$ results in faster runtimes. 

Turning to Figure \ref{fig:EF1-PO1} we see much of the same performance. Again the MILP implementations outperform the dynamic programming solutions as we scale up the number of agents and items. In the EF1/PROP1 setting we do see a slightly more pronounced difference in runtime based on $\phi$ with less correlated preferences leading to faster runtimes. 

In comparing Figure \ref{fig:EF-PO} and \ref{fig:EF1-PO1} we can make some general statements about the reletive performance of EF/PROP versus EF1/PROP1. First, it is interesting to note that when $n < 5$, the dynamic programming algorithms are faster for both EF/EF1 and PROP/PROP1 settings. Secondly, for both cases, and more so for the EF1/PROP1, we can clearly see that for both the MILP solution and the dynamic programming solution, PROP/PROP1 is an easier problem than EF/EF1, resulting in faster runtimes.

Finally, we tracked the number of instances that had solutions that admitted either an EF or PROP allocation. Of the 900 instances we generated, only 11.2\% admitted an EF solution while 71.3\% admitted a PROP solution. Whereas, when looking at the EF1/PROP1 relaxations, that number is 100\% for both since we know these allocations always exist. This can be seen as evidence for the importance of studying the relaxation of the solution concepts to EF1 and PROP1 since many instances do not admit fully EF or PROP solutions.

\section{Conclusions and Future Work}
We provide a clear picture of the computational complexity of computing allocations that are both fair and maximize the utilitarian social welfare.  
We find that while existence can be decided efficiently when we have $n=2$ agents, in most cases both the question of existence for UM and fair allocations as well as finding UM within fair allocations, is computationally hard.  However, we are able to demonstrate positive results in the form of pseudopolynomial time algorithms when the number of agents is a constant $n \geq 3$ for both the fairness concepts EF1 and PROP1 as well as their stronger counterparts EF and PROP.

Although EF1 is stronger than PROP1, 
an algorithm for UM-within-EF1 does not imply an algorithm for  UM-within-PROP1, since the maximum utilitarian welfare in a PROP1 allocation might be higher than the maximum utilitarian welfare in an EF1 allocation. This raises an interesting question of how much utilitarian welfare is lost when moving from PROP1 to EF1 allocations. Our algorithms allow to study this question empirically in future work.

By using different thresholds in  the algorithms of 
Theorems \ref{um-within-prop-pseudopoly}
and \ref{um-within-prop1-pseudopoly},
these algorithms can be adapted to handle other fairness notions, such as weighted proportionality (for agents with different entitlements), maximin-share fairness \citep{Budi11a} in cases in which the maximin-shares of the agents are known,
or thresholds computed from picking sequences, recently introduced by  \citet{gourves2021fairness}.

Our other algorithms can be adapted to handle other notions of fairness that are based on the ``up to 1 item'' concept.
This is because the tables constructed by these algorithms contain information about the bundle values and about items that are contained / not contained in them.

Similarly, our algorithms can be adapted to handle \emph{egalitarian optimality} (maximizing the minimum utility), by using the minimum in the last step instead of the sum.
However, we do not know if they 
can be adapted to more complex efficiency notion, such as \emph{rank maximality} \citep{irving2006rank}.

Other directions for future work include:
\begin{enumerate}
\item Improving the run-time complexity of the algorithms for a fixed number of agents;
\item Devising faster algorithms for restricted domains, and for approximately-maximal utilitarian welfare;
\item Constructing datasets of utilities based on experiments with humans, and examining the performance of our algorithms on such datasets \cite{DBLP:conf/ijcai/Mattei20}.
\end{enumerate}

%

\section*{Acknowledgments}
We are grateful to Andrzej Kaczmarczyk for his help in understanding his paper \cite{bredereck2019high}, and to anonymous reviewers for their helpful comments.

Haris Aziz is supported by  the Australian Defence Science and Technology Group (DSTG) under the project
``Auctioning for distributed multi vehicle planning'' (MYIP 9953) and by the Asian Office of Aerospace Research and Development (AOARD) under the project 
``Efficient and fair decentralized task allocation algorithms for autonomous vehicles'' (FA2386-20-1-4063). 
Xin Huang is supported in part at the Technion by an Aly Kaufman Fellowship.
Nicholas Mattei is supported by NSF Award IIS-2007955, ``Modeling and Learning Ethical Principles for Embedding into Group Decision Support Systems,'' and an IBM Faculty Research Award.
Erel Segal-Halevi is supported by the Israel Science Foundation (grant no. 712/20).

\clearpage
\appendix
\renewcommand{\thesection}{\Alph{section}}
\normalsize
\section*{APPENDIX}
\section{EF1 vs. PROP1}
\label{sec:ef1-prop1}
It is known that, with additive valuations, EF1 implies PROP1. For completeness, we provide a proof below.
\begin{proposition}
Let $i\in N$ be an agent with an additive utility function $u_i$. If an allocation $p$ is EF1 for $i$, then $p$ is PROP1 for $i$.
\end{proposition}
\begin{proof}
If $p(i) = O$ then the allocation if obviously PROP1 for $i$, so we assume $p(i) \subsetneq O$.
Let $U := \max_{o\in O\setminus p(i)} u_i(o) = $ the value of the most valuable item that is not allocated to $i$.

The definition of EF1 implies that, for every agent $j\neq i$:
\begin{align*}
u_i(p(i)) \geq u_i(p(j)) - U.
\end{align*}
The same obviously holds when $j = i$. Summing over all $j\in \{1,\ldots,n\}$ yields:
\begin{align*}
n\cdot u_i(p(i)) \geq u_i(O) - n\cdot U.
\end{align*}
Dividing by $n$ yields:
\begin{align*}
& 
u_i(p(i)) \geq u_i(O)/n - U
\\
\implies &
u_i(p(i)) + U \geq u_i(O)/n
\end{align*}
which is the condition for PROP1.
\end{proof}

The following example shows that EF1 is strictly stronger than PROP1, even when there are only two agents.
\begin{example}
 	Consider the following instance with 2 agents and 7 items where the number is the utility to each agent.
 	\begin{center}
 	\begin{tabular}{ccccccc}
 	\hline
 	 &$a$&$b$ ($\times 6$)&\\
 	\hline
 	Alice:  & 4& 1 \\
 	\hline
 	Bob:  & 4 &1 \\
 	\hline
 	\end{tabular}
 	\end{center}
 	Suppose Alice has $a$ and Bob has all the six $b$ items.
 	Then the allocation is PROP1 but not EF1 for Alice.
\end{example}

\end{document}